\documentclass[runningheads,12pt]{llncs}
\usepackage[a4paper,margin=20mm]{geometry}
\usepackage{graphicx}

\usepackage[english]{babel}
\usepackage{amsfonts}
\usepackage{inputenc}
\usepackage[ruled]{algorithm}
\usepackage{color}
\usepackage{algpseudocode}
\usepackage{booktabs,multicol,multirow}
\usepackage{fancyhdr}
\usepackage{epic}
\usepackage{xspace}
\usepackage{amssymb,amsmath,verbatim}
\usepackage{multirow}
\usepackage{multicol}
\usepackage{adjustbox}
\usepackage{threeparttable}
\usepackage{float}

\usepackage{amsthm}
\usepackage{chngcntr}
\newtheorem{thm}{Theorem}

\newtheorem{lem}[thm]{Lemma}
\newtheorem{rem}[thm]{Remark}

\def\ZZ{\mathbb{Z}}
\def\QQ{\mathbb{Q}}
\def\RR{\mathbb{R}}

\def\cal{\mathcal}
\def\bf{\mathbf}

\def\TrapGen{\mathsf{TrapGen}}
\def\PK{\mathsf{PK}}
\def\SK{\mathsf{SK}}
\def\CT{\mathsf{CT}}
\def\SampleLeft{\mathsf{SampleLeft}}
\def\SampleRight{\mathsf{SampleRight}}
\def\SampleBasisLeft{\mathsf{SampleBasisLeft}}
\def\SampleP{\mathsf{SampleP}}
\def\SamplePre{\mathsf{SamplePre}}
\def\td{\mathsf{td}}

\def\Setup{\mathsf{Setup}}

\def\Enc{\mathsf{Enc}}
\def\Dec{\mathsf{Dec}}
\def\Td{\mathsf{Td}}
\def\Test{\mathsf{Test}}
\def\Pr{\mathrm{Pr}}
\def\Adv{\mathsf{Adv}}
\def\OW{\textsf{OW-CPA}}
\def\IND{\textsf{IND-CPA}}
\def\OWCCA{\textsf{OW-CCA2}}
\def\INDCCA{\textsf{IND-CCA2}}

\def\a{\alpha}
\def\u{\bf{u}}
\def\v{\bf{v}}
\def\t{\theta}
\def\e{\bf{e}}
\def\d{\bf{d}}

\def\x{\bf{x}}

\def\L{\Lambda}
\def\Lp{\Lambda^{\perp}}
\def\b{\bf{b}}
\def\s{\bf{s}}
\def\c{\bf{c}}

\def\ID{\mathsf{ID}}
\def\PP{\mathsf{PP}}

\def\RLWE{\textbf{Ring-LWE}_{n,q,D_{R,\sigma}}}
\def\RSIS{\textbf{Ring-SIS}_{q,n,m,\beta}}
\begin{document}
	\title{CCA2-secure Lattice-based Public Key Encryption with Equality Test in Standard Model}
	\titlerunning{}
	\authorrunning{}
	\author{Dung Hoang Duong\inst{1}\and Partha Sarathi Roy\inst{1}\and Willy Susilo\inst{1}\and\\  Kazuhide Fukushima\inst{2}\and Shinsaku Kiyomoto\inst{2}  \and Arnaud Sipasseuth\inst{2}}
	\institute{Institute of Cybersecurity and Cryptology,\\ School of Computing and Information Technology, University of Wollongong\\
		Northfields Avenue, Wollongong NSW 2522, Australia\\
		\email{\{hduong,partha,wsusilo\}@uow.edu.au} 
		\and
		Information Security Laboratory, KDDI Research, Inc.\\
		2-1-15 Ohara, Fujimino-shi, Saitama, 356-8502, Japan\\
		\email{\{ka-fukushima,kiyomoto,ar-sipasseuth\}@kddi-research.jp}
	}

	\maketitle              
	
	\begin{abstract}
		With the rapid growth of cloud storage and cloud computing services, many organisations and users choose to store the data on a cloud server for saving costs. However, due to security concerns, data of users would be encrypted before sending to the cloud. However, this hinders  a problem of computation on encrypted data in the cloud, especially in the case of performing data matching in various medical scenarios.
		Public key encryption with equality test (PKEET) is a powerful tool that allows the authorized cloud server to check whether two ciphertexts are generated by the same message. PKEET has then become a  promising candidate for many practical applications like efficient data management on encrypted databases.
		Lee et al. (Information Sciences 2020) proposed a generic construction of PKEET schemes in the standard model and hence it is possible to yield the first instantiation of post-quantum PKEET schemes based on lattices. At ACISP 2019, Duong et al. proposed a direct construction of PKEET over integer lattices in the standard model. However, their scheme does not reach the CCA2-security. 		
		In this paper, we propose an efficient CCA2-secure PKEET scheme  based on ideal lattices. In addition, we present a modification of the scheme by Duong et al. over integer lattices to attain the CCA2-security. Both schemes are proven secure in the standard model, and they enjoy the security in the upcoming quantum computer era.
	\end{abstract}
	%
	%
	%
\section{Introduction}
With the rapid growth of cloud computing, more and more organizations and individuals tend to store their data in cloud as well as outsouce their heavy computations to the cloud services. Since the data is normally sensitive, e.g., medical records of patients, it is desired to encrypt the data before sending or outsourcing to the cloud services. However, this causes a big problem for doing computations on encrypted data, especially in performing data matching in various medical scenarios. 	

Public key encryption with equality test (PKEET), introduced by Yang et al.~\cite{Yang10}, is a special kind of public key encryption that allows any authorised tester with a given trapdoor to test whether two ciphertexts are generated by the same message. This special feature has made PKEET a powerful tool utilized  in many practical applications, such as keyword search on encrypted data, encrypted data partitioning for efficient encrypted data management, personal health record systems, and spam filtering in encrypted email systems.
Since then, there have been an intensive research in this direction 
with the appearance of improvements and ones with additional functionalities~\cite{Tang11,Tang12,Tang12b,Ma15,Lee16}. However, those schemes are proven to be secure in the random oracle model, which is not a realistic, even though no insecurity has been found in practical schemes~\cite{CanettiGH98}. It is a desire to construct cryptographic schemes, e.g., PKEET, in the standard model.

Up to the present, there are only a few PKEET schemes in the standard model. Lee et al.~\cite{Lee2016} first proposed a generic construction of a PKEET scheme. Their method is to use a $2$-level hierarchical identity-based encryption (HIBE) scheme together with a one-time signature scheme. The HIBE scheme is used for generating an encryption scheme and for equality test, and the signature scheme is used for making the scheme CCA2-secure, based on the method of transforming an identity-based encryption (IBE) scheme to a CCA2-secure encryption scheme of Canetti et al~\cite{CHK04}. As a result, they obtain a CCA2-secure PKEET scheme given that the underlying HIBE scheme is IND-sID-CPA secure and the one-time signature scheme is strongly unforgeable. From their generic construction, it is possible to obtain a PKEET in standard model under many hard assumptions via instantiations.
In another recent paper, Zhang et al.~\cite{Zhang2019} proposed a direct construction of a CCA2-secure PKEET scheme based on pairings without employing strong cryptographic primitives such as HIBE schemes and strongly secure signatures as the generic construction of Lee et al.~\cite{Lee2016}. Their technique comes from a CCA2-secure public key encryption scheme by~\cite{Lai10} which was directly constructed by an idea from IBE. A comparison with an instantiation from Lee et al.~\cite{Lee2016} on pairings shows that their direct construction is much more efficient than the instantiated one.

All aforementioned existing schemes base their security on the hardness of some number-theoretic assumptions which are insecure against the quantum computer attacks~\cite{Shor97}. The generic construction by Lee et al.~\cite{Lee2016} is the first one with the possibility of yielding a post-quantum instantiation based on lattices, since lattice cryptography is the only  post-quantum cryptography area up to present offers  HIBE primitives, e.g., ~\cite{ABB10-EuroCrypt}. At ACISP 2019, Duong et al.~\cite{Duong19} proposed a direct PKEET in standard model based on lattices from IBE scheme by Agrawal et al.~\cite{ABB10-EuroCrypt}. However, their scheme is not CCA2-secure as claimed. 

\medskip
\noindent\textbf{Our contribution:}
In this paper, we propose an efficient PKEET scheme based on ideal lattices. The core construction is to utlize the IBE scheme by Agrawal et al.~\cite{ABB10-EuroCrypt} in the ideal lattice version proposed by Bert et al.~\cite{BFRS18}. In order to achieve the CCA2-security, we apply the generic CHK transformation by Caneti et al.~\cite{CHK04} in which we employ the efficient one-time strong signature scheme over ideal lattices by Lyubashevsky and Micciancio~\cite{LM18}. As a result, we obtain an efficient CCA2-secure PKEET scheme in standard model over ideal lattices. The security of the scheme is reduced to the hardness of the learning with errors (LWE) problem and the short integer solution (SIS) problem over rings; see Section~\ref{sec:construction-ideal} for the detail.

We next revisit and modify the PKEET construction over integer lattices by Duong et al.~\cite{Duong19} to achieve CCA2-security by correctly applying the CHK transformation. In order to reserve he construction as in~\cite{Duong19}, we utilize the strong signature scheme in~\cite{ABB10-EuroCrypt}, in which we can use the public key as a verification key. As a trade-off, the ciphertext needs to add one more matrix in $\mathbb{Z}_q^{n\times m}$ for verificaiton and a signature in $\mathbb{Z}_q^{2m}$. This results to a CCA2-secure PKEET scheme, which is still 
more efficient than the generic construction of Lee et al.~\cite{Lee2016}; see Section~\ref{sec:Duong-revisit} for the detail. We also present in the Appendix the instantiation of the construction by Lee et al.~\cite{Lee2016}. As a result, our scheme in Section~\ref{sec:Duong-revisit} has much smaller ciphertext size compared to that of Lee et al.	
Note that in both construction, one just needs to generate a one-time signature in the encryption process which in turn reduces the ciphertext size. We also note that the PKEET version over integer lattices of our proposed scheme in Section~\ref{sec:construction-ideal}, by utilizing the IBE scheme in~{\cite{ABB10-EuroCrypt}} and the one-time signature scheme from SIS in~\cite{LM18}, is still more efficient than the revisited scheme in Section~\ref{sec:Duong-revisit} and the instantiation of the construction by Lee et al.~\cite{Lee2016} in the Appendix.



\section{Preliminaries}
\subsection{Public key encryption with equality test (PKEET)}\label{sec:security model}
In this section, we will recall the model of PKEET and its security model.

We remark that a PKEET system is a multi-user setting. Hence we assume that in our system throughout the paper, each user is assigned with an index $i$ with $1\leq i\leq N$ where $N$ is the number of users in the system.

\begin{definition}[PKEET] \label{def:PKEET}
	Public key encryption with equality test (PKEET) consists of the following polynomial-time algorithms:
	\begin{itemize}
		\item $\Setup(\lambda)$: On input a security parameter $\lambda$ and set of parameters, it outputs the a pair of a user's public key $\PK$ and secret key $\SK$.
		\item $\Enc(\PK,\bf{m})$: On input the public key $\PK$ and a message $\bf{m}$, it outputs a ciphertext $\CT$.
		\item $\Dec(\SK,\CT)$: On input the secret key $\SK$ and a ciphertext $\CT$, it outputs a message $\bf{m}'$ or $\perp$.
		\item $\Td(\SK_i)$: On input the secret key $\SK_i$ for the user $U_i$, it outputs a trapdoor $\td_i$.
		\item $\Test(\td_i,\td_j,\CT_i,\CT_j)$: On input two trapdoors $\td_i, \td_j$ and two ciphertexts $\CT_i, \CT_j$ for users $U_i$ and $U_j$ respectively, it outputs $1$ or $0$.
	\end{itemize}    
\end{definition}

\noindent\textbf{Correctness.} We say that a PKEET scheme is \textit{correct} if the following three condition hold:
\begin{description}
	\item[(1)] For any security parameter $\lambda$, any user $U_i$ and any message $\bf{m}$, it holds that
	$$\small{\Pr\left[ {\begin{gathered}
				\Dec(\SK_i,\CT_i)=\bf{m}\end{gathered}  
			\left| \begin{gathered}
				(\PK_i,\SK_i)\gets\Setup(\lambda)\\
				\CT_i\gets\Enc(\PK_i,\bf{m})
			\end{gathered}  \right.} \right]=1}.$$
	\item[(2)] For any security parameter $\lambda$, any users $U_i$, $U_j$ and any messages $\bf{m}_i, \bf{m}_j$, it holds that:    
	$$\small{\Pr\left[{
			\Test\left( \begin{gathered}
				\td_i \\
				\td_j \\
				\CT_i \\
				\CT_j \\ 
			\end{gathered}  \right) = 1\left| \begin{array}{l}
				(\PK_i,\SK_i)\gets\Setup(\lambda) \\
				\CT_i\gets\Enc(\PK_i,\bf{m}_i) \\
				\td_i\gets\Td(\SK_i) \\
				(\PK_j,\SK_j)\gets\Setup(\lambda) \\
				\CT_j\gets\Enc(\PK_j,\bf{m}_j) \\
				\td_j\gets\Td(\SK_j) 
			\end{array}  \right.} \right]=1}$$
	if $\bf{m}_i=\bf{m}_j$ regardless of whether $i=j$.
	
	\item[(3)] For any security parameter $\lambda$, any users $U_i$, $U_j$ and any messages $\bf{m}_i, \bf{m}_j$, it holds that
	$$\small{\Pr\left[{
			\Test\left( \begin{gathered}
				\td_i \\
				\td_j \\
				\CT_i \\
				\CT_j \\ 
			\end{gathered}  \right) = 1\left| \begin{array}{l}
				(\PK_i,\SK_i)\gets\Setup(\lambda) \\
				\CT_i\gets\Enc(\PK_i,\bf{m}_i) \\
				\td_i\gets\Td(\SK_i) \\
				(\PK_j,\SK_j)\gets\Setup(\lambda) \\
				\CT_j\gets\Enc(\PK_j,\bf{m}_j) \\
				\td_j\gets\Td(\SK_j) 
			\end{array}  \right.} \right]}$$
	is negligible in $\lambda$ for any ciphertexts $\CT_i$, $\CT_j$ such that $\Dec(\SK_i,\CT_i)\ne\Dec(\SK_j,\CT_j)$ regardless of whether $i=j$.
\end{description}

\noindent\textbf{Security model of PKEET.} For the security model of PKEET, we consider two types of adversaries:
\begin{itemize}
	\item[$\bullet$] Type-I adversary: for this type, the adversary can request to issue a trapdoor for the target user and thus can perform equality tests on the challenge ciphertext. The aim of this type of adversaries is to reveal the message in the challenge ciphertext.
	\item[$\bullet$] Type-II adversary: for this type, the adversary cannot request to issue a trapdoor for the target user and thus cannot perform equality tests on the challenge ciphertext. The aim of this type of adversaries is to distinguish which message is in the challenge ciphertext between two candidates.
\end{itemize}
The security model of a PKEET scheme against two types of adversaries above is described in the following.\\

\noindent\textbf{OW-CCA2 security against Type-I adversaries.} We illustrate the game between a challenger $\cal{C}$ and a Type-I adversary $\cal{A}$ who can have a trapdoor for all ciphertexts of the target user, say $U_{\theta}$, that he wants to attack, as follows:
\begin{enumerate}
	\item \textbf{Setup:} The challenger $\cal{C}$ runs $\Setup(\lambda)$ to generate the key pairs $(\PK_i,\SK_i)$ for all users with $i=1,\cdots,N$, and gives $\{\PK_i\}_{i=1}^N$ to $\cal{A}$.
	\item \textbf{Phase 1:}  The adversary $\cal{A}$ may make queries polynomially many times adaptively and in any order to the following oracles:
	\begin{itemize}
		\item $\cal{O}^{\SK}$: an oracle that on input an index $i$ (different from $\theta$), returns the $U_i$'s secret key $\SK_i$.
		\item $\cal{O}^\Dec$: an oracle that on input a pair of an index $i$ and a ciphertext $\CT_i$, returns the output of $\Dec(\SK_i,\CT_i)$ using the secret key of the user $U_i$.
		\item $\cal{O}^\Td$: an oracle that on input an index $i$, return $\td_i$ by running $\td_i\gets\Td(\SK_i)$ using the secret key $\SK_i$ of the user $U_i$.
	\end{itemize}
	\item \textbf{Challenge:} $\cal{C}$ chooses a random message $\bf{m}$ in the message space and run $\CT_{\theta}^*\gets\Enc(\PK_{\theta},\bf{m})$, and sends $\CT_{\theta}^*$ to $\cal{A}$.
	\item \textbf{Phase 2:} $\cal{A}$ can query as in Phase $1$ with the following constraints:
	\begin{itemize}
		\item The index $\theta$ cannot be queried to the key generation oracle $\cal{O}^{\SK}$;
		\item The pair of the index $\theta$ and the ciphertext $\CT_{\theta}^*$ cannot be queried to the decryption oracle $\cal{O}^\Dec$.
	\end{itemize}
	\item \textbf{Guess:} $\cal{A}$ output $\bf{m}'$.
\end{enumerate}
The adversary $\cal{A}$ wins the above game if $\bf{m}=\bf{m}'$ and the success probability of $\cal{A}$ is defined as
$$\Adv^{\OW}_{\cal{A},\text{PKEET}}(\lambda):=\Pr[\bf{m}=\bf{m}'].$$
\begin{rem}
	If the message space is polynomial in the security parameter or the min-entropy of the message distribution is much lower than the security parameter then a Type-I adversary $\cal{A}$ with a trapdoor for the challenge ciphertext can reveal the message in polynomial-time or small exponential time in the security parameter, by performing the equality tests with the challenge ciphertext and all other ciphertexts of all messages generated by himself. Hence to prevent this attack, we assume that the size of the message space $\cal{M}$ is exponential in the security parameter and the min-entropy of the message distribution is sufficiently higher than the security parameter.
\end{rem}

\noindent\textbf{IND-CCA2 security against Type-II adversaries.} We present the game between a challenger $\cal{C}$ and a Type-II adversary $\cal{A}$ who cannot have a trapdoor for all ciphertexts of the target user $U_{\theta}$ as follows:
\begin{enumerate}
	\item \textbf{Setup:} The challenger $\cal{C}$ runs $\Setup(\lambda)$ to generate the key pairs $(\PK_i,\SK_i)$ for all users with $i=1,\cdots,N$, and gives $\{\PK_i\}_{i=1}^N$ to $\cal{A}$.
	\item \textbf{Phase 1:}  The adversary $\cal{A}$ may make queries polynomially many times adaptively and in any order to the following oracles:
	\begin{itemize}
		\item $\cal{O}^{\SK}$: an oracle that on input an index $i$ (different from $t$), returns the $U_i$'s secret key $\SK_i$.
		\item $\cal{O}^\Dec$: an oracle that on input a pair of an index $i$ and a ciphertext $\CT_i$, returns the output of $\Dec(\SK_i,\CT_i)$ using the secret key of the user $U_i$.
		\item $\cal{O}^\Td$: an oracle that on input an index $i$ (different from $t$), return $\td_i$ by running $\td_i\gets\Td(\SK_i)$ using the secret key $\SK_i$ of the user $U_i$.
	\end{itemize}
	\item \textbf{Challenge:} $\cal{A}$ chooses two messages $\bf{m}_0$ $\bf{m}_1$ of same length and pass to $\cal{C}$, who then selects a random bit $b\in\{0,1\}$, runs $\CT^*_{\theta, b}\gets\Enc(\PK_{\theta},\bf{m}_b)$ and sends $\CT^*_{\theta,b}$ to $\cal{A}$.
	\item \textbf{Phase 2:} $\cal{A}$ can query as in Phase $1$ with the following constraints:
	\begin{itemize}
		\item The index $t$ cannot be queried to the key generation oracle $\cal{O}^{\SK}$ and the trapdoor generation oracle $\cal{O}^\Td$;
		\item The pair of the index $\theta$ and the ciphertext $\CT_{\theta,b}^*$ cannot be queried to the decryption oracle $\cal{O}^\Dec$.
	\end{itemize}
	\item \textbf{Guess:} $\cal{A}$ output $b'$.
\end{enumerate}
The adversary $\cal{A}$ wins the above game if $b=b'$ and the advantage of $\cal{A}$ is defined as
$$\Adv_{\cal{A},\text{PKEET}}^{\IND}:=\left|\Pr[b=b']-\frac{1}{2}\right|.$$

\subsection{Lattices}
Throughout the paper, we will mainly focus on integer lattices, which are discrete subgroups of $\ZZ^m$. Specially, a lattice $\Lambda$ in $\ZZ^m$ with basis $B=[\b_1,\cdots,\b_n]\in\ZZ^{m\times n}$, where each $\b_i$ is written in column form, is defined as
$$\Lambda:=\left\{\sum_{i=1}^n\b_ix_i | x_i\in\ZZ~\forall i=1,\cdots,n \right\}\subseteq\ZZ^m.$$
We call $n$ the rank of $\L$ and if $n=m$ we say that $\L$ is a full rank lattice. In this paper, we mainly consider full rank lattices containing $q\ZZ^m$, called $q$-ary lattices, defined as the following, for a given matrix $A\in\ZZ^{n\times m}$ and $\bf{u}\in\ZZ_q^n$
\begin{align*}
	\L_q(A) &:= \left\{ \e\in\ZZ^m ~\rm{s.t.}~ \exists \bf{s}\in\ZZ_q^n~\rm{where}~A^T\bf{s}=\bf{e}\mod q \right\}\\
	\Lp_q(A) &:= \left\{ \e\in\ZZ^m~\rm{s.t.}~A\e=0\mod q \right\} \\
	\L_q^{\bf{u}}(A) &:=  \left\{ \e\in\ZZ^m~\rm{s.t.}~A\e=\bf{u}\mod q \right\}
\end{align*}
Note that if $\bf{t}\in\L_q^{\bf{u}}(A)$ then $\L_q^{\bf{u}}(A)=\Lp_q(A)+\bf{t}$.

Let $S=\{\s_1,\cdots,\s_k\}$ be a set of vectors in $\mathbb{R}^m$. We denote by $\|S\|:=\max_i\|\s_i\|$ for $i=1,\cdots,k$, the maximum $l_2$ length of the vectors in $S$. We also denote $\tilde{S}:=\{\tilde{\s}_1,\cdots,\tilde{\s}_k \}$ the Gram-Schmidt orthogonalization of the vectors $\s_1,\cdots,\s_k$ in that order. We refer to $\|\tilde{S}\|$ the Gram-Schmidt norm of $S$.

Ajtai~\cite{Ajtai99} first proposed how to sample a uniform matrix $A\in\ZZ_q^{n\times m}$ with an associated basis $S_A$ of $\Lp_q(A)$ with low Gram-Schmidt norm. It is improved later by Alwen and Peikert~\cite{AP09} in the following Theorem.

\begin{thm}\label{thm:TrapGen}
	Let $q\geq 3$ be odd and $m:=\lceil 6n\log q\rceil$. There is a probabilistic polynomial-time algorithm $\TrapGen(q,n)$ that outputs a pair $(A\in\ZZ_q^{n\times m},S\in\ZZ^{m\times m})$ such that $A$ is statistically close to a uniform matrix in $\ZZ_q^{n\times m}$ and $S$ is a basis for $\Lp_q(A)$ satisfying
	\[\|\tilde{S}\|\leq O(\sqrt{n\log q})\quad\text{and}\quad\|S\|\leq O(n\log q)\]
	with all but negligible probability in $n$.
\end{thm}

\begin{definition}[Gaussian distribution]
	Let $\L\subseteq\ZZ^m$ be a lattice. For a vector $\bf{c}\in\RR^m$ and a positive parameter $\sigma\in\RR$, define:
	$$\rho_{\sigma,\c}(\x)=\exp\left(\pi\frac{\|\x-\c\|^2}{\sigma^2}\right)\quad\text{and}\quad
	\rho_{\sigma,\c}(\L)=\sum_{\x\in\L}\rho_{\sigma,\c}(\x).    $$
	The discrete Gaussian distribution over $\L$ with center $\c$ and parameter $\sigma$ is
	$$\forall \bf{y}\in\L\quad,\quad\cal{D}_{\L,\sigma,\c}(\bf{y})=\frac{\rho_{\sigma,\c}(\bf{y})}{\rho_{\sigma,\c}(\L)}.$$
\end{definition}
For convenience, we will denote by $\rho_\sigma$ and $\cal{D}_{\L.\sigma}$ for $\rho_{\bf{0},\sigma}$ and $\cal{D}_{\L,\sigma,\bf{0}}$ respectively. When $\sigma=1$ we will write $\rho$ instead of $\rho_1$. We recall below in Theorem~\ref{thm:Gauss} some useful results. The first one is from~\cite{CHKP10} and formulated in ~{\cite[Theorem 17]{ABB10-EuroCrypt}}. The second one from~{\cite[Theorem 19]{ABB10-EuroCrypt}}   and the last one is from~{\cite[Corollary 30]{ABB10-EuroCrypt}}.

\begin{thm}\label{thm:Gauss}
	Let $q> 2$ and let $A, B$ be a matrix in $\ZZ_q^{n\times m}$ with $m>n$ and $B$ is rank $n$. Let $T_A, T_B$ be a basis for $\Lp_q(A)$ and  $\Lp_q(B)$ respectively.
	Then for $c\in\RR^m$ and $U\in\ZZ_q^{n\times t}$:
	\begin{enumerate}
		\item Let $M$ be a matrix in $\ZZ_q^{n\times m_1}$ and $\sigma\geq\|\widetilde{T_A}\|\omega(\sqrt{\log(m+m_1)})$. Then there exists a PPT algorithm $\SampleLeft(A,M,T_A,U,\sigma)$ that outputs a matrix $\e\in\ZZ^{(m+m_1)\times t}$ distributed statistically close to $\cal{D}_{\L_q^{U}(F_1),\sigma}$ where $F_1:=(A~|~M)$. In particular $\e\in \L_q^{U}(F_1)$, i.e., $F_1\cdot\e=U\mod q$.
		\item Let $M$ be a matrix in $\ZZ_q^{n\times m_1}$ and $\sigma\geq\|\widetilde{T_A}\|\omega(\sqrt{\log(m+m_1)})$. Then there exists a PPT algorithm $\SampleBasisLeft(A,M,T_A,\sigma)$ that outputs a basis $T_{F_1}$ for $\Lp_q(F_1)$ where $F_1:=(A~|~M)$, provided that $A$ is of rank $n$.
		
		\item Let $R$ be a matrix in $\ZZ^{k\times m}$ and let $s_R:=\sup_{\|\x\|=1}\|R\x\|$. Let $F_2:=(A~|~AR+B)$. Then for  $\sigma\geq\|\widetilde{T_B}\|s_R\omega(\sqrt{\log m})$, there exists a PPT algorithm \\$\SampleRight(A,B,R,T_B,U,\sigma)$ that outputs a matrix $\e\in\ZZ^{(m+k)\times t}$ distributed statistically close to $\cal{D}_{\L_q^{U}(F_2),\sigma}$. In particular $\e\in \L_q^{\u}(F_2)$, i.e., $F_2\cdot\e=U\mod q$. 
		
		Note that when $R$ is a random matrix in $\{-1,1\}^{m\times m}$ then $s_R<O(\sqrt{m})$ with overwhelming probability (cf.~{\cite[Lemma 15]{ABB10-EuroCrypt}}).
		
	\end{enumerate}
\end{thm}

The security of our construction reduces to the LWE (Learning With Errors) problem introduced by Regev~\cite{Regev05}.
\begin{definition}[LWE problem]
	Consider publicly a prime $q$, a positive integer $n$, and a distribution $\chi$ over $\ZZ_q$. An $(\ZZ_q,n,\chi)$-LWE problem instance consists of access to an unspecified challenge oracle $\cal{O}$, being either a noisy pseudorandom sampler $\cal{O}_\s$ associated with a secret $\s\in\ZZ_q^n$, or a truly random sampler $\cal{O}_\$$ who behaviors are as follows:
	\begin{description}
		\item[$\cal{O}_\s$:] samples of the form $(\u_i,v_i)=(\u_i,\u_i^T\s+x_i)\in\ZZ_q^n\times\ZZ_q$ where $\s\in\ZZ_q^n$ is a uniform secret key, $\u_i\in\ZZ_q^n$ is uniform and $x_i\in\ZZ_q$ is a noise withdrawn from $\chi$.
		\item[$\cal{O}_\$$:] samples are uniform pairs in $\ZZ_q^n\times\ZZ_q$.
	\end{description}\
	The $(\ZZ_q,n,\chi)$-LWE problem allows responds queries to the challenge oracle $\cal{O}$. We say that an algorithm $\cal{A}$ decides the $(\ZZ_q,n,\chi)$-LWE problem if 
	$$\Adv_{\cal{A}}^{\mathsf{LWE}}:=\left|\Pr[\cal{A}^{\cal{O}_\s}=1] - \Pr[\cal{A}^{\cal{O}_\$}=1] \right|$$    
	is non-negligible for a random $\s\in\ZZ_q^n$.
\end{definition}
Regev~\cite{Regev05} showed that (see Theorem~\ref{thm:LWE} below) when $\chi$ is the distribution $\overline{\Psi}_\alpha$ of the random variable $\lfloor qX\rceil\mod q$ where $\alpha\in(0,1)$ and $X$ is a normal random variable with mean $0$ and standard deviation $\alpha/\sqrt{2\pi}$ then the LWE problem is hard. 
\begin{thm}\label{thm:LWE}
	If there exists an efficient, possibly quantum, algorithm for deciding the $(\ZZ_q,n,\overline{\Psi}_\alpha)$-LWE problem for $q>2\sqrt{n}/\alpha$ then there is an efficient quantum algorithm for approximating the SIVP and GapSVP problems, to within $\tilde{\cal{O}}(n/\alpha)$ factors in the $\|.\|_2$ norm, in the worst case.
\end{thm}
Hence if we assume  the hardness of approximating the SIVP and GapSVP problems in lattices of dimension $n$ to within polynomial (in $n$) factors, then it follows from Theorem~\ref{thm:LWE} that deciding the LWE problem is hard when $n/\alpha$ is a polynomial in $n$.

In this paper, we also reduce the security of our scheme to the hardness of the Short Integer Solution (SIS) problem, stated as the following. For SIS problem that we use in this paper, we will deal with infinite norm $\|.\|_\infty$.

\begin{definition}[SIS problem]
	Consider publicly a prime $q$, positive integers $n,m$ and a positive real $\beta$. An $\text{SIS}^\infty_{n,q,\beta,m}$ problem instance is as follows: given a matrix $\bf{A}\in\ZZ_q^{n\times m}$, find a nonzero vector $\bf{s}\in\ZZ^m$ of norm $\|\bf{s}\|_\infty\leq \beta$ such that $\bf{A}\bf{s}=\bf{0}\mod q$.
\end{definition}

Hardness of SIS has first started by the seminal work of Ajtai~\cite{Ajtai96} and is subsequently improved by a series of work~\cite{MR04,GPV08,MP13}. In summary, we have the following.

\begin{thm}[{\cite[Theorem 2.7]{LM18}}]\label{thm:SIS}
	For any $\beta>0$ and any sufficiently large $q\geq\beta\cdot\sqrt{m}n^{\Omega(1)}$ with at most $n^{\Omega(1)}$ factors less than $\beta$, solving $\text{SIS}^\infty_{n,q,\beta,m}$ problem (on average, with nonnegligible probability $n^{-\Omega(1)}$) is at least as hard as solving SIVP in the worst case on any $n$-dimensional lattice within a factor $\gamma=\max\{1,\beta^2\sqrt{m}/q\}\cdot\title{O}(\beta\sqrt{nm})$.
	
	In particular, for any constant $\epsilon>0,\beta\leq n^\epsilon$, and $q\geq\beta\sqrt{m}n^\epsilon$, $\text{SIS}^\infty_{n,q,\beta,m}$  is hard on average under the assumption that SIVP is hard in the worst case for $\gamma=\tilde{O}(\beta\sqrt{nm})$.
\end{thm}

\subsection{Ideal lattices}\label{sec:ideal lattices}
In this paper, we construct efficient schemes based on ideal lattices, i.e., lattices arising from polynomial rings. Specially, when $n$ is a power of two, we consider the ring $R:=\ZZ[x]/(x^n+1)$, the ring of integers of the cyclotomic number field $\QQ[x]/(x^n+1)$. The ring $R$ is isomorphic to the integer lattice $\ZZ^n$ through mapping a polynomial $f=\sum_{i=0}^{n-1}f_ix^i$ to its vector of coefficients $(f_0,f_1,\cdots,f_{n-1})$ in $\ZZ^n$. In this paper, we will consider such ring $R$ and denote by $R_q=R/qR=\ZZ_q[x]/(x^n+1)$ where $q=1\mod 2n$ is a prime. In this section, we follow~\cite{BertFRS18-implement} to recall some useful results in ideal lattices. 

We use ring variants of LWE, proposed by~\cite{LPR10}, and proven to be as hard as the SIVP problems on ideal lattices.

\begin{definition}[$\RLWE$] Given $\bf{a}=(a_1,\cdots,a_m)^T\in R_q^m$ a vector of $m$ uniformly random polynomials, and $\bf{b}=\bf{a}s+\bf{e}$ where $s\hookleftarrow U(R_q)$ and $\bf{e}\hookleftarrow D_{R^m,\sigma}$, distinguish $(\bf{a},\bf{b}=\bf{a}s+\bf{e})$ from $(\bf{a},\bf{b})$ drawn from the uniform distribution over $R_q^m\times R_q^m$.
\end{definition}

\begin{thm}[\cite{LPR13}]\label{thm:RLWE}
	For any $m=\text{poly}(n)$, cyclotomic ring (e.g., $R=\ZZ[x]/(x^n+1)$ with $n$ a power of $2$) of degree $n$ (over $\mathbb{Z}$), and appropriate choices of modulus $q$ and error distribution $D_{R,\alpha q}$ of error rate $\alpha<1$, solving the $\RLWE$ problem (with $\sigma=\alpha q$) is at least as hard as quantumly solving the $\text{SVP}_\gamma$ problem on arbitrary ideal lattices in $R$, for some $\gamma=\text{poly}(n)/\alpha$. 
\end{thm}

We also use ring variants of SIS which are proven to be as hard as the SVP problems on ideal lattices.

\begin{definition}[$\RSIS$] Given a matrix $H\in R_q^{1\times m}$, find a non-zero vector $\bf{s}\in R_q^m$ such that $\|\bf{s}\|_\infty\leq\beta$ and $H\bf{s}=0\mod\! q$.
\end{definition}

\begin{thm}[\cite{LM06}]
	For $m>\log q/\log(2\beta), \gamma=16\beta\cdot m\cdot n\log^2n$, and $q\geq\frac{\gamma\cdot\sqrt{n}}{4\log n}$, solving the $\RSIS$ problem in uniformly random matrices in $R_q^{1\times m}$ is at least as hard as solving $\text{SVP}_\gamma^\infty$ in any ideal in the ring $\ZZ[x]/(x^n+1)$.
\end{thm}

In this paper, we use the ring version of trapdoors for ideal lattices, introduced in~\cite{MP12} and recently improved in~\cite{GM18}.
\begin{definition}
	Define $\bf{g}=(1,2,4,\cdots,2^{k-1})^T\in R_q$ with $k=\lceil\log_2(q)\rceil$ and call $\bf{g}$ the \textbf{gadget vector}, i.e., for which the inversion of $f_{\bf{g}^T}(\bf{z})=\bf{g}^T\bf{z}\in R_q$ is easy. The lattice $\L_q^\perp(\bf{g}^T)$ has a publicly known basis $B_q\in R^{k\times k}$ which satisfies that $\|\widetilde{B}_q\|\leq\sqrt{5}$.
\end{definition}

\begin{definition}[$\bf{g}$-trapdoor]
	Let $\bf{a}\in R_q^m$ and $\bf{g}\in R_q^k$ with $k=\lceil\log_2(q)\rceil$ and $m>k$. A $\bf{g}$-trapdoor for $\bf{a}$ consist in a matrix of small polynomials $T\in R^{(m-k)\times k}$, following a discrete Gaussian distribution of parameter $\sigma$, such that 
	$$\bf{a}^T\left(\begin{array}{c}T \\ I_k \end{array} \right) =h\bf{g}^T$$
	for some invertible element $h\in R_q$. The polynomial $h$ is called the tag associated to $T$. The quality of the trapdoor is measured by its largest singular value $s_1(T)$.
\end{definition}
The Algorithm~\ref{algo:TrapGen} shows how to generate a random vector $\bf{a}\in R_q^m$ together with its trapdoor $T$.

\begin{algorithm}
	\begin{itemize}
		\item{Input:} the ring modulus $q$, a Gaussian parameter $\sigma$, optional $\bf{a}'\in R_q^{m-k}$ and $h\in R_q$. If no $\bf{a}', h$ are given, the algorithm chooses $\bf{a}'\hookleftarrow U(R_q^{m-k})$ and $h=1$.
		\item{Output:} $\bf{a}\in R_q^m$ with its trapdoor $T\in R^{(m-k)\times k}$, of norm $\| T\|\leq t\sigma\sqrt{(m-k)n}$ associated to the tag $h$.
	\end{itemize}
	\begin{enumerate}
		\item Choose $T\hookleftarrow D_{R^{(m-k)\times k},\sigma}, \bf{a}=(\bf{a}'^T|h\bf{g}-\bf{a}'^TT	)^T$.
		\item Return $(\bf{a},T)$.
	\end{enumerate}
	\caption{Algorithm  $\TrapGen(q,\sigma,\bf{a}',h)$}
	\label{algo:TrapGen}
\end{algorithm}

One of the main algorithm we use in our scheme is the preimage sampling algorithm $\SamplePre$, illustrated in Algorithm~\ref{algo:SamplePre}, which finds $\bf{x}$ such that $f_{\bf{a}^T}(\bf{x})=u$ for a given $u\in R_q$ and a public $\bf{a}\in R_q^m$ using the $\bf{g}$-trapdoor $T$ of $\bf{a}$, where $(\bf{a},T)\gets\TrapGen(q,\sigma,\bf{a}',h)$ as in Algorithm~\ref{algo:TrapGen}.

\begin{algorithm}\label{algo:SamplePre}
	\begin{itemize}
		\item{Input:} $\bf{a}\in R_q^m$, with its trapdoor $T\in R^{(m-k)\times k}$ associated to an invertible tag $h\in R_q$, $u\in R_q$ and $\zeta,\sigma,\alpha$ three Gaussian parameters
		\item{Output} $\bf{x}\in R_q^m$ following a discrete Gaussian distribution of parameters $\zeta$ satisfying $\bf{a}^T\bf{x}=u\in R_q$.
	\end{itemize}
	\begin{enumerate}
		\item $\bf{p}\gets\mathsf{SampleP}(q,\zeta,\alpha,T)$, $v\gets h^{-1}(u-\bf{a}^T\bf{p})$
		
		\item $\bf{z}\gets\mathsf{SamplePolyG}(\sigma,v)$, $\bf{x}\gets \bf{p}+\left(\begin{array}{c}T \\ I_k \end{array} \right)\bf{z}$
	\end{enumerate}
	\caption{Algorithm  $\mathsf{SamplePre}(T,\bf{a},h,\zeta,\sigma,\alpha,u)$}
\end{algorithm}

The algorithm~$\SamplePre$ uses the following two algorithm:
\begin{itemize}
	\item The algorithm $\SampleP(q,\zeta,\alpha,T)$ on input the ring modulus $q$, $\zeta$ and $\alpha$ two Gaussian parameters and $T\hookleftarrow D_{R^{(m-k)\times k},\sigma}$, output $\bf{p}\hookleftarrow D_{R^m,\sqrt{\Sigma_{\bf{p}}}}$ where $\Sigma_\bf{p}=\zeta^2I_m-\alpha^2\left(\begin{array}{c}
		T \\ I_k
	\end{array}\right)(T^T I_k)$.
	\item The algorithm $\mathsf{SamplePolyG}$ on input a Gaussian parameter $\sigma$ and a target $v\in R_q$, outputs $\bf{z}\hookleftarrow D_{\Lp_q(\bf{g}^T),\alpha,v}$ with $\alpha=\sqrt{5}\sigma$.
\end{itemize}
For the special case of the cyclotomic number ring $R=\ZZ[x]/(x^n+1)$ in our paper, the algorithm ~$\SampleP$ can be efficiently implemented as in~{\cite[Section 4.1]{GM18}}. For algorithm~$\mathsf{SamplePolyG}$, one needs to call the algorithm~$\mathsf{SampleG}$ in~{\cite[Section 3.2]{GM18}} $n$ times.

Our construction follows the IBE construction in~\cite{ABB10-EuroCrypt}. In such a case, we need an encoding hash function $H:\ZZ_q^n\to R_q$ to map identities to $\ZZ_q^n$ to invertible elements in $R_q$; such an $H$ is called \textit{encoding with Full-Rank Differences (FRD)} in \cite{ABB10-EuroCrypt}. We require that $H$ satisfies the following properties:
\begin{itemize}
	\item for all distinct $u,v\in\ZZ_q^n$, the element $H(u)-H(v)\in R_q$ is invertible; and
	\item $H$ is computable in polynomial time (in $n\log(q)$).
\end{itemize}
To implement such an encoding $H$, there are several methods proposed in~\cite{GM18,DM14,LS18} and we refer to {\cite[Section 2.4]{GM18}} for more details.

\subsection{Strong one-time signature}	\label{sec:one-time signature}
In this paper, we utilize the one-time signature proposed by Lyubashevsky and Micciancio in~\cite{LM18}. The signature is proven to be strongly unforgeable under the hardness of SIS/Ring-SIS problem.
\subsubsection{Strong one-time signature from SIS}
The scheme is parameterized by
\begin{itemize}
	\item integers $m,k,n,q,b,w$
	\item $\mathcal{H}=\mathbb{Z}_q^{n\times m}$,
	\item $\mathcal{K} = \{\mathbf{K}\in\ZZ_q^{m\times k}: \|\mathbf{K}\|_\infty\leq b \}$,
	\item $\mathcal{M}\subseteq\{\mathbf{m}\in\{0,1\}^k: \|\bf{m}\|_1=w \}$,
	\item $\mathcal{S} =\{ \mathbf{s}\in\ZZ_q^m : \|\mathbf{s}\|_\infty\leq wb\}$.
\end{itemize}
The scheme consists of the following four algorithms.
\begin{itemize}
	\item Setup: A random matrix $H\in\cal{H}$ is chosen and can be shared by all users.
	\item Key Generation: A secret key $\mathbf{K}\in\mathcal{K}$ is chosen uniformly at random. Output the public key $\mathbf{K}'=\mathbf{H}\mathbf{K}\in\ZZ_q^{n\times k}$.
	\item Signing: Given a message $\bf{m}\in\mathcal{M}$, and the secret key $\mathbf{K}$, output a signature $\mathbf{s}=\mathbf{K}\mathbf{m}\in\ZZ_q^m$.
	\item Verification: Given a message and signature pair $(\mathbf{m},\mathbf{s})$, checks if $\mathbf{s}\in\mathcal{S}$ and $\mathbf{H}\mathbf{s}=\mathbf{K}'\mathbf{m}$. 
\end{itemize}
It is easy to check the correctness of the scheme. We have the following.
\begin{thm}[\cite{LM18}]\label{sig strong SIS}
	For any $\epsilon>0$, let $q\geq 2wb\sqrt{m}n^\epsilon$ and $b=\lceil\frac{q^{n/m}2^{\lambda/m}-1}{2}\rceil$. Then the one-time signature scheme above is strongly unforgeable under the assumption that $\text{SIVP}_\gamma$ is hard in the worst case for $\gamma=\tilde{O}(wb\sqrt{nm})\max\{1,2wb/n^\epsilon\}$.
	
	In particular, for $m=\lceil(\lambda+n\log_2q)/\log_23\rceil$, $b=1$ and $q\geq 2w\sqrt{m}n^\epsilon$, the scheme is strongly unforgeable under the assumption that $\text{SIVP}_\gamma$ is hard in the worst case for $\gamma=\tilde{O}(w\sqrt{nm})\max\{1,2w/n^\epsilon\}$.
	
\end{thm}

\subsubsection{Strong one-time signature from Ring-SIS}
The scheme is parameterized by
\begin{itemize}
	\item integers $m,k,n$,
	\item a ring $R_q=\ZZ_q[x]/(x^n+1)$
	\item $\cal{H}= R_q^{1\times m}$
	\item $\cal{K} =\{[\bf{k}_1,\bf{k}_2]\in R_q^{m\times 2}: \|\bf{k}_1\|_\infty\leq b, \|\bf{k}_2\|_\infty\leq wb\}$
	\item $\cal{M}\subseteq\{\bf{m}=[m_1,1]^T\in R_q^2, \|m_1\|_\infty\leq 1, \|m_1\|_1\leq w \}$
	\item $\cal{S}=\{\bf{s}\in R_q^m: \|\bf{s}\|_\infty \leq 2wb \}$.
\end{itemize}

The scheme consists of four algorithms described as follows.
\begin{itemize}
	\item Setup: A random matrix $H\in\cal{H}$ is chosen and can be shared by all users.
	\item Key Generation: A secret key $K\in\cal{K}$ is chosen uniformly at random. Output the public key $\hat{K}=HK\in R_q^{1\times 2}$.
	\item Signing: Given a message $\bf{m}\in\cal{M}$ and the secret key $K$, output a signature $\bf{s}=K\bf{m}\in R_q^m$.
	\item Verification: Given a message-signature pair $(\bf{m},\bf{s})$, checks if $\bf{s}\in\cal{S}$ and $H\bf{s}=\hat{K}\bf{m}$.
\end{itemize}

The correctness of the scheme is clear. We have the following.
\begin{thm}[\cite{LM18}]\label{sig strong RSIS}
	Let $b=\lfloor (|\cal{M}|^{1/n}2^{\lambda/n}q)^{1/m}\rceil$ and $q>8wb$. Then the one-time signature described above is strongly unforgeable based on the assumed average-case hardness of $\textbf{Ring-SIS}_{n,m,q,4wb}$ problem. Furthermore, for $\gamma=64wbmn\log^2n$ and $q\geq\frac{\gamma\sqrt{n}}{4\log n}$, the scheme is secure based on the worst-case hardness of $\text{SVP}_\gamma^\infty$ in all $n$-dimensional ideals of the ring $\ZZ[x]/(x^n+1)$.
\end{thm}

When basing the problem on the worst-case hardness of SVP on ideal lattices, we can set $w=O(n/\log n), b=1, m=O(n\log n)$, modulus $q=n^{2.5}\log n$ and the worst-case approximation factor $\gamma=O(n^2\log^2n)$.


\section{PKEET over Ideal Lattices}\label{sec:construction-ideal}
In this section, we propose a CCA2-secure PKEET over ideal lattices. The scheme is inherited from the one in.
Our scheme is presented in Section~\ref{sec:construction ideal}. The correctness of the scheme and the choice of parameters are presented in Section~\ref{sec:Params-ideal} while the security analysis is presented in Section~\ref{sec: security ideal}.  
\subsection{Construction}\label{sec:construction ideal}
The proposed PKEET consists of the following algorithms.

\medskip
\noindent $\Setup(1^n)$:
On input the security parameter $1^n$, do the following:
\begin{enumerate}
	\item Generate $\bf{a}\in R_q^m$ and its trapdoor $T_\bf{a}\in R^{(m-k)\times k}$ by $(\bf{a},T_\bf{a})\gets\TrapGen(q,\sigma,h=0)$, i.e., 
	$\bf{a}=((\bf{a}')^{T} | -(\bf{a}')^{T}T_\bf{a})^T$
	\item Generate $\bf{b}\in R_q^m$ and its trapdoor $T_\bf{b}\in R^{(m-k)\times k}$ by $(\bf{b},T_\bf{b})\gets\TrapGen(q,\sigma,h=0)$, i.e., 
	$\bf{b}=((\bf{b}')^{T} | -(\bf{b}')^{T}T_\bf{b})^T$
	\item Sample uniformly random $u\hookleftarrow U(R_q)$.
	\item Output $\PK = (\bf{a},\bf{b},u)\in R_q^{2m+1}$ and $\SK=(T_\bf{a},T_\bf{b})$
\end{enumerate} 

\medskip
\noindent $\Enc(\PK,M)$]:
Given a message $M\in R_2$, do the following
\begin{enumerate}
	\item Sample $s_1,s_2\hookleftarrow U(R_q)$, $e'_1,e'_2\hookleftarrow D_{R,\tau}$ and compute
	\begin{align*}
		\CT_1 &= u\cdot s_1+e'_1 +M\cdot\lfloor q/2\rfloor\in R_q,\\
		\CT_2&=u\cdot s_2 +e'_2 + H(M)\cdot\lfloor q/2\rfloor\in R_q.
	\end{align*}
	
	where $H$ is a hash function mapping from $\{ 0,1\}^*$ to the message space $\cal{M}$.
	\item Choose $\bf{k}\in R^{(m-k)\times 2}$ and compute $\bf{v}=\bf{a'}^T\cdot\bf{k}\in R^{2}$.
	\item Compute ${h}=H_1(\bf{v})\in R_q$, where $H_1$ is a hash function that maps $\{0,1\}^*$ to invertible elements in $R_q$.
	\item Compute $\bf{a}_h = \bf{a}^T + (\bf{0} | h\bf{g})^T = ((\bf{a}')^T | h\bf{g}-(\bf{a}')^TT_\bf{a})^T$.
	\item Compute $\bf{b}_h = \bf{b}^T + (\bf{0} | h\bf{g})^T = ((\bf{b}')^T | h\bf{g}-(\bf{b}')^TT_\bf{b})^T$.
	\item Choose $\bf{y},\bf{y}'\hookleftarrow D_{R^{m-k},\tau}, \bf{z}, \bf{z}'\in D_{R^k,\gamma}$ and compute
	\begin{align*}
		\CT_3 &= \bf{a}_h\cdot s_1 +(\bf{y}^T,\bf{z}^T)^T\in R_q^m,\\
		\CT_4 &= \bf{b}_h\cdot s_2 +((\bf{y}')^T,(\bf{z}')^T)^T\in R_q^m.
	\end{align*}
	\item Compute ${m}=H_2(\CT_1\|\CT_2\|\CT_3\|\CT_4)$ and $\sigma=\bf{k}\cdot[{m},1]^T$. Here $H_2$ is a function that maps from $\{0,1\}^*$ into the space $\{{m}\in R : \|m\|_\infty\leq 1, \|m\|_1\leq \delta \}$ with $\delta=O(n/\log n)$.
	
	\item Output the ciphertext
	$$\CT=(\sigma,\bf{v},\CT_1,\CT_2,\CT_3,\CT_4)\in R_q^{3m+4-k}.$$
\end{enumerate}
\medskip
\noindent Decrypt($\SK$,$\CT$):$\\$
On input the secret key $\SK=(T_\bf{a},T_\bf{b})$ and a ciphertext $\CT=(\sigma,\bf{v},\CT_1,\CT_2,\CT_3,\CT_4)$, do the following:
\begin{enumerate}
	\item Compute ${m}=H_2(\CT_1\|\CT_2\|\CT_3\|\CT_4)$ and check whether $\bf{a'}^T\cdot\sigma = \bf{v}\cdot[{m},1]$. If not then output $\perp$, otherwise continue.
	\item Compute $h=H_1(\bf{v})$ and construct $\bf{a}_h$ and $\bf{b}_h$ as in Step 3-4 in the Encryption process.
	\item Sample short vectors $\bf{x},\bf{x}'\in R_q^m$:
	\begin{align*}
		\bf{x} &\gets\SamplePre(T_\bf{a},\bf{a}_h,h,\zeta,\sigma,\alpha,u) \\
		\bf{x}' &\gets\SamplePre(T_\bf{b},\bf{b}_h,h,\zeta,\sigma,\alpha,u)
	\end{align*}
	\item Compute $\bf{w}=\CT_1-\CT_3^T\bf{x}\in R_q$.
	\item For each $w_i$, if it is closer to $\lfloor q/2\rfloor$ than to $0$, then output $M_i=1$, otherwise $M_i=0$. Then we obtain the message $M$.
	\item Compute $\bf{w}'=\CT_2-\CT_4^T\bf{x}'\in R_q$.
	\item For each $w'_i$, if it is closer to $\lfloor q/2\rfloor$ than to $0$, then output $\bf{h}_i=1$, otherwise $\bf{h}_i=0$. Then we obtain an element $\bf{h}$.
	\item If $\bf{h} = H(M)$ then output $M$, otherwise output $\perp$.
\end{enumerate}

Let $U_i$ and $U_j$ be two users of the system. We denote by $\CT_i = (\sigma_i,\bf{v}_i,\CT_{i,1},\CT_{i,2},\CT_{i,3},\CT_{i,4})$ (resp. $\CT_j = (\sigma_j,\bf{v}_j,\CT_{j,1},\CT_{j,2},\CT_{j,3},\CT_{j,4})$) be a ciphertext of $U_i$ (resp. $U_j$).

\medskip
\noindent Trapdoor($\SK_i$): On input a user $U_i$'s secret key $\SK_i = (T_{i,\bf{a}}, T_{i,\bf{b}})$, it outputs a trapdoor $\td_{1,i}=T_{i,\bf{b}}$.

\medskip
\noindent Test($\td_{1,i},\td_{1,j},\CT_i,\CT_j$):
On input trapdoors $\td_{1,i},\td_{1,j}$ and ciphertexts $\CT_i,\CT_j$ for users $U_i, U_j$ respectively, computes
\begin{enumerate}
	\item For each $i$ (resp. $j$), do the following
	\begin{enumerate}
		\item Compute $h_i=H(v_i)$ and sample $\bf{x'}_i\in R_q^m$ from
		$$\bf{x'}_i\gets\SamplePre(T_{i,\bf{b}},\bf{b}_h,h_i,\zeta,\sigma,\alpha,u).$$
		\item Compute $\bf{w}_i=\CT_{i,2}-\CT_{i,4}^T\bf{x'}_i$. For each $k=1,\cdots,n$, if $w_{i,k}$ is closer to  $\lfloor q/2\rfloor$ than to $0$, then output $\bf{h}_{ik}=1$, otherwise $\bf{h}_{ik}=0$. Then we obtain the element $\bf{h}_i$ (resp. $\bf{h}_j$).	
	\end{enumerate}
	\item Output $1$ if $\bf{h}_i=\bf{h}_j$, and $0$ otherwise.
\end{enumerate}

\begin{lem}[Correctness]
	With the choice of parameters as in~\ref{sec:Params-ideal}, our proposed PKEET  is correct,  assuming that the hash function $H$ is collision-resitant.
\end{lem}

\begin{proof}
	Let $\bf{x}=(\bf{x}_0^T|\bf{x}_1^T)^T$ with $\bf{x}_0\in R_q^{m-k}$ and $\bf{x}_1\in R_q^k$. To correctly decrypt a ciphertext, we need the error term $e_1'-(\bf{y}^T|\bf{z}^T)(\bf{x}_0^T|\bf{x}_1^T)^T = e_1'-\bf{y}^T\bf{x}_0-\bf{z}^T\bf{x}_1$ to be bounded by $\lfloor q/4\rfloor$, which is satisfied by the choice of parameters in~Section~\ref{sec:Params-ideal}. Similarly, for the test procedure, one needs to correctly decrypt $H(M)$ and the equality test works correctly given that $H$ is collision-resistant.
\end{proof}

\subsection{Correctness and Parameters}\label{sec:Params-ideal}

We follow~{\cite[Section 4.2]{BertFRS18-implement}} and~{\cite[Section 4.2]{LM18}}\footnote{This choice of parameters ensures the strongly unforgeability of the signature scheme in Section~\ref{sec:one-time signature}; cf.~Theorem~\ref{sig strong RSIS}.} for choosing parameters for our scheme as the following.
\begin{enumerate}
	\item The modulus $q$ is chosen to be a prime and $q=n^{2.5}\log n$.
	\item We choose $m-k=2$ and $m=O(n\log n)$
	\item The Gaussian parameter $\sigma$ for the trapdoor sampling is $\sigma>\sqrt{(\ln(2n/\epsilon)/\pi)}$ (\cite{MP12}) where $n$ is the maximum length of the ring polynomials, and $\epsilon$ is the desired bound on the statistical error introduced by each randomized rounding operation. This parameter is also chosen to ensure the hardness of $\RLWE$ problem.
	\item The Gaussian parameter $\sigma$ for the $G$-sampling is $\alpha = \sqrt{5}\sigma$ (\cite{MP12}).
	\item The parameter $\zeta$ is chosen such that $\zeta>\sqrt{5}C\sigma^2(\sqrt{kn}+\sqrt{2n}+t')$ for $C\cong 1/\sqrt{2\pi}$ and $t'\geq 0$, following~\cite{GM18}.
	\item For decrypting correctly, we need
	$$
	t\tau\sqrt{n}+2t^2\tau\zeta n+t^2\gamma\zeta k n<\lfloor q/4\rfloor.
	$$
	\item Finally, we choose $\mu=t\sigma\tau\sqrt{2n}$ and $\gamma=2t\sigma\tau\sqrt{n}$ so that $\gamma$ satisfies $\gamma^2=(\sigma\|\bf{e}'\|)^2+\mu^2$.
	
	\item The parameter $t$ here is chosen such that a vector $\bf{x}$ sampled in $D_{\ZZ^m,\sigma}$ has norm $\|\bf{x}\|\leq t\sigma\sqrt{m}$. Note that
	$$\Pr_{x\hookleftarrow D_{\ZZ,\sigma}}[|x|>t\sigma]\leq\mathrm{erfc}(t/\sqrt{2})$$
	with $\mathrm{erfc}(x)=1-\frac{2}{\pi}\int_{0}^x\exp^{-t^2}dt$. One can choose, for example, $t=12$ (see {\cite[Section 2]{BertFRS18-implement}}).
\end{enumerate}

\subsection{Security Analysis}\label{sec: security ideal}
In this Section, we prove the proposed PKEET scheme achieves $\OWCCA$ and $\INDCCA$ security under the hardness of $\RLWE$ and $\RSIS$ problems.

\begin{thm}[$\OWCCA$]\label{thm:OWCCA Ideal}
	The proposed PKEET in Section~\ref{sec:construction ideal} with parameters as in Section~\ref{sec:Params-ideal} is $\OWCCA$ secure provided that $H_1$ is a collision-resistant hash function and the $\RLWE$ and $\RSIS$ problems are hard.
\end{thm}

\begin{proof}
	The proof is proceeded through a sequence of games between an adversary $\cal{A}$ and an $\OWCCA$ challenger starting from \textbf{Game 0}, the original $\OWCCA$ game in Section~\ref{sec:security model}. In the last game, \textbf{Game 2}, the challenge ciphertext is chosen uniformly at random. Hence, the advantage of the adversary in this game is zero. We will show that the following games are indistinguishable, conditioned on the one-wayness of the hash function $H$ and the hardness of the $\RLWE$ and $\RSIS$ problems. which then will imply that the advantage of the adversary against the system is negligible.
	
	\medskip
	\noindent \textbf{Game 0:} This is the original $\INDCCA$ game.
	
	\medskip
	\noindent \textbf{Game 1:} In this game, we change the way how the challanger generates the public key for the user with index $\theta$ and the challange ciphertext $\CT^*_{\theta}$. At the start of the experiment, choose random $\bf{k}^*,\bf{a}'\in R_q^{m-k}$ and let  the public parameter $\bf{a}$ generated by $\TrapGen(q,\sigma,\bf{a}',-h^*_{\theta})$ where $h^*_{\theta}=H_1(\bf{v}^*)$ with $\bf{v}^*=\bf{a'}^T\cdot\bf{k}$. Hence, the public parameter is $\bf{a}=((\bf{a}')^T | -h^*_{\theta}\bf{g}-(\bf{a}')^TT_\bf{a})^T$, where the first part $\bf{a}'\in R_q^{m-k}$ is chosen from the uniform distribution. For the second part $\bf{a}'^TT_\bf{a}=(\sum_{i=1}^{m-k}a_it_{i,1},\cdots,\sum_{i=1}^{m-k}a_it_{i,k})$ is indistinguishable from the uniform distribution.
	In our paper, we choose $m-k=2$ and $\bf{a}'=(1,a)$ with $a\hookleftarrow U(R_q)$ and the public key $\bf{a}=(1,a|-(at_{2,1}+t_{1,1}),\cdots,-(at_{2,k}+t_{1,k}))$ looks uniform followed by the $\RLWE$ assumption, given that the secret and error follow the same distribution. The remainder of the game is unchanged and similar to \textbf{Game 0}.
	
	Note that whenever $\cal{A}$ queries $\cal{O}^{\Dec}(\t,\CT_\t)$ with $\CT_\t=(\sigma,\bf{v},\CT_{\t,1},\CT_{\t,2},$ $\CT_{\t,3},\CT_{\t,4})$ then $\cal{B}$ does as follows.
	\begin{itemize}
		\item If $\bf{v}=0$ then $\cal{B}$ aborts. It happens with negligible probability.
		\item If $(\bf{v},\CT_{\t,1},\CT_{\t,2},$ $\CT_{\t,3},\CT_{\t,4})=(\bf{v}^*,\CT^*_{\t,1},\CT^*_{\t,2},$ $\CT^*_{\t,3},\CT^*_{\t,4})$ and $\sigma\ne\sigma^*$, then we can use $\cal{A}$ to break the strongly unforgeable one-time signature in Section~\ref{sec:one-time signature}, which is impossible by the hardness of $\RSIS$ problem.
		\item If $(\CT_{\t,1},\CT_{\t,2},$ $\CT_{\t,3},\CT_{\t,4})=(\CT^*_{\t,1},\CT^*_{\t,2},$ $\CT^*_{\t,3},\CT^*_{\t,4})$ but $\bf{v}\ne\bf{v}^*$, then this implies that $H_1(\bf{v})=h^*_\t=H_1(\bf{v}^*)$. Hence we break the collision-resistance of the hash function $H_1$.
		\item Otherwise, $\cal{B}$ can answer as usual using the trapdoor $T_\bf{a}$, except if $H_1(\bf{v})=h^*_\t$, which happens with probability at most the advantage $\epsilon_{H_1,\mathsf{CR}}$ of breaking the collision-resistance of $H_1$.
	\end{itemize}
	It follows that \textbf{Game 1} and \textbf{Game 0} are indistinguishable.

	\medskip
	\noindent \textbf{Game 2:} In this game, we change how the challenge ciphertext is generated. The ciphertext $\CT^*_{\theta}$ is now chosen uniformly in $R_q^{3m+4-k}$. It is obvious that the advantage of the adversary $\cal{A}$ in this game is zero.

	We now show that \textbf{Game 2} and \textbf{Game 1} are indistinguishable for $\cal{A}$ by doing a reduction from $\RLWE$ problem. 
	
	Now $\cal{B}$ receives $m-k+1$ samples $(a_i,b_i)_{0\leq i\leq m-k}$ as an instance of the decisional RLWE problem. Let $\bf{a}'=(a_1,\cdots,a_{m-k})^T\in R_q^{m-k}$ and $\bf{b}'=(b_1,\cdots,b_{m-k})^T\in R_q^{m-k}$. The simulator runs $\TrapGen(q,\sigma,\bf{a}',-h^*_\theta)$, and we get $\bf{a}=((\bf{a}')^T | -h^*_{\theta}\bf{g}-(\bf{a}')^TT_{\bf{a}})^T$ as in \textbf{Game 1}. 
Construct $\bf{b}$ as in \textbf{Game 1}. Next $\cal{B}$ set $u=a_0$ and sends $\PK_\theta=(\bf{a},\bf{b},u)$ to $\cal{A}$ as the public key of the user $\theta$. 
	
	At the challenge phase, the simulator chooses a message $M$ and computes the challenge ciphertext $\CT^*_\theta\gets\Enc(\PK_\theta,M)$ as follows:
	
	\begin{enumerate}
		\item Set $\CT^*_{\theta,1}\gets b_0 + M\cdot\lfloor q/2\rfloor.$
		\item Choose a uniformly random $s_2\in R_q$ and $e_2'\hookleftarrow D_{R,\tau}$ and compute
		$$\CT^*_{\theta,2} = u\cdot s_2 + e_2' + H'(M)\cdot\lfloor q/2\rfloor\in R_q.$$

		\item Set 
		$$\CT^*_{\t,3}=\left[ 
		\begin{array}{c}
			\bf{b}' \\
			-\bf{b}'T_\bf{a} +\widehat{\bf{e}}
		\end{array}
		\right]\in R_q^m$$
		with $\widehat{\bf{e}}\hookleftarrow D_{R_q^k,\,u}$ for some real $\mu$.
		\item Choose $\bf{y}'\hookleftarrow D_{R^{m-k},\tau}$, $\bf{z}'\hookleftarrow R_{R^k,\gamma}$ and set
		$$\CT^*_{\t,4}=\bf{b}_{h_\t}\cdot s_2+((\bf{y}')^T,(\bf{z}')^T)^T\in R_q^m.$$
		\item Compute $m=H_2(\CT^*_{\t,1}\| \CT^*_{\t,2}\|\CT^*_{\t,3}\|\CT^*_{\t,4})$ and $\sigma^*=\bf{k}^*\cdot[m,1]^T$.
	\end{enumerate}		
	Then $\cal{B}$ sends $\CT^*_\t=(\sigma^*,\bf{v}^*,\CT^*_{\t,1},\CT^*_{\t,2},\CT^*_{\t,3},\CT^*_{\t,4})$ to $\cal{A}$.
	
	When the samples $(a_i,b_i)$ are LWE samples, then $\bf{b}'=\bf{a}'s_1+\bf{e}'$ and $b_0=a_0s_1+e_0$ for some $s_1\in R_q$ and $e_0\hookleftarrow D_{R,\tau}$, $\bf{e}'\hookleftarrow D_{R^{m-k},\tau}$. It implies that 
	$$
	\CT^*_{\t,1} = u\cdot s_1+e_0 + M\cdot\lfloor q/2\rfloor ~~~ and$$
	$$\CT^*_{\t,3} = \bf{a}_{h_\t}\cdot s_1 + (\bf{e}'^T | \bf{z}^T) $$
	where $\bf{z}=-\bf{e}'^TT_{\bf{a}}+\widehat{\bf{e}}^T$ is indistinguishable from a sample drawn from the distribution $D_{R^k,\gamma}$ with $\gamma^2=(\sigma\|\bf{e}'\|)^2+\mu^2$ for $\mu$ well chosen. 
	
	Then $\CT^*_{\t}$ is a valid ciphertext.
	
	When the $(a_i,b_i)$ are uniformly random in $R_q^2$, then obviously $\CT^*_{\t}$ also looks uniform.
	
	$\cal{A}$ guesses if it is interacting with \textbf{Game 2} or \textbf{Game 1}. The simulator outputs the final guess as the answer to the
	RLWE problem. It follows that the advantage of $\cal{A}$ is negligible due to the hardness of the  $\RLWE$ problem. This completes the proof.
\end{proof}

\begin{thm}[$\INDCCA$]\label{thm:INDCCA Ideal}
	The proposed PKEET in Section~\ref{sec:construction ideal} with parameters as in Section~\ref{sec:Params-ideal} is $\INDCCA$ secure provided that $H_1$ is a collision-resistant hash function, $H$ is a one-way hash function, and that the $\RLWE$ problem and the $\RSIS$ problem are hard.
\end{thm} 

\begin{proof}
	The proof is proceeded through a sequence of games between an adversary $\cal{A}$ and an $\INDCCA$ challenger starting from the original $\INDCCA$ game in Section~\ref{sec:security model}.  It is then required to show that the games are indistinguishable, conditioned on the one-wayness of the hash function $H$, the collision resistance of the hash function $H_1$ and the hardness of the $\RLWE$ and $\RSIS$ problems. The proof is similar to that of Theorem~\ref{thm:OWCCA Ideal} and we will omit it here. 
	Note that in this scenario, the adversary $\cal{A}$ can have the trapdoor for the target user to do the equality test. In such case, the adversary can follow the testing procedure to obtain the hash value $H(m_i)$ of the challenge ciphertext. If he can break the one-wayness of $H$ then he can win the game.
\end{proof}


\section{Revisting PKEET over integer lattices from~\cite{Duong19} }\label{sec:Duong-revisit}
In this paper, we modify the construction of PKEET from~\cite{Duong19} to achieve CCA2-security. We apply the CHK transformation~\cite{CHK04} with the strongly one-time signature from~\cite{LM18} described in Section~\ref{sec:one-time signature}. Note that the version of PKEET in Section~\ref{sec:construction-ideal} over integer lattices is much more efficient, in terms of key size, than the one presented in this paper, which we modify directly on the scheme proposed by Duong et al.~\cite{Duong19}.
\subsection{Construction}

\medskip
\noindent Setup($\lambda$): On input a security parameter $\lambda$, set the parameters $q,n,m,\sigma,\alpha$ as in section \ref{sec:params}
\begin{enumerate}
	\item Use $\TrapGen(q,n)$ to generate uniformly random $n\times m$-matrices $A, A'\in\ZZ_q^{n\times m}$ together with trapdoors $T_{A}$ and $T_{A'}$ respectively.
	\item Select $l+1$ uniformly random $n\times m$ matrices $A_1,\cdots,A_l,B\in\ZZ_q^{n\times m}$.
	\item Let $H: \{0,1\}^*\to \{0,1\}^t$ and $H':\{0,1\}^*\to \{-1,1\}^l$ be hash functions.
	\item Select a uniformly random matrix $U\in\ZZ_q^{n\times t}$.
	\item Output the public key and the secret key
	$$\PK=(A,A',A_1,\cdots,A_l,B,U), \SK=(T_A,T_{A'}).$$
\end{enumerate} 

\medskip
\noindent Encrypt($\PK,\bf{m}$): On input the public key $\PK$ and a message $\bf{m}\in\{0,1\}^t$, do:
\begin{enumerate}
	\item Choose a uniformly random $\bf{s}_1, \bf{s}_2\in\ZZ_q^n$
	\item Choose $\bf{x}_1,\bf{x}_2\in\overline{\Psi}_\alpha^t$ and compute\footnote{Note that for a message $\bf{m}\in\{0,1\}^t$, we choose a random binary string $\bf{m}'$ of fixed length $t'$ large enough and by abusing of notation, we write $H(\bf{m})$ for $H(\bf{m}'\|\bf{m})$.}
	\begin{align*}
		\bf{c}_1 &= U^T\bf{s}_1 +\bf{x}_1 +\bf{m}\big\lfloor\frac{q}{2}\big\rfloor\in\ZZ_q^t,\\
		\bf{c}_2 &= U^T\bf{s}_2 +\bf{x}_2 +H(\bf{m})\big\lfloor\frac{q}{2}\big\rfloor \in\ZZ_q^t.
	\end{align*}
	
	\item Choose $K\in\ZZ_q^{m\times k}$ uniformly at random and compute $D=AK\in\ZZ_q^{n\times k}$.
	\item Compute $\bf{b}=H'(\bf{c}_1\|\bf{c}_2\|D)\in\{-1,1\}^l$, and set 
	$$F_1=(A|B+\sum_{i=1}^lb_iA_i), F_2=(A'|B+\sum_{i=1}^lb_iA_i).$$
	\item Choose $l$ uniformly random matrices $R_i\in\{-1,1\}^{m\times m}$ for $i=1,\cdots,l$ and define $R=\sum_{i=1}^lb_iR_i\in\{-l,\cdots,l\}^{m\times m}$.
	\item Choose $\bf{y}_1, \bf{y}_2\in\overline{\Psi}_\alpha^m$ and set $\bf{z}_1=R^T\bf{y}_1, \bf{z}_2=R^T\bf{y}_2\in\ZZ_q^m$.
	\item Compute
	$$\bf{c}_3=F_1^T\bf{s}_1+[\bf{y}_1^T|\bf{z}_1^T]^T, \bf{c}_4=F_2^T\bf{s}_2+[\bf{y}_2^T|\bf{z}_2^T]^T\in\ZZ_q^{2m}.$$
	
	\item Compute $\bf{d}=H'(\bf{c}_1\|\bf{c}_2\|\bf{c}_3\|\bf{c}_4)\in\{0,1\}^k$ such that $\|\bf{d}\|_1\leq w$ and compute $\bf{u}=K\bf{d}\in\ZZ_q^{m}$.
	\item The ciphertext is $$\CT=(\bf{c}_1,\bf{c}_2,\bf{c}_3,\bf{c}_4,\bf{u},D)\in\ZZ_q^{2t+5m}\times\ZZ_q^{n\times k}.$$    
\end{enumerate}

\medskip
\noindent Decrypt($\PK,\SK,\CT$): On input public key $\PK$, private key $\SK$ and a ciphertext $\CT=(\bf{c}_1,\bf{c}_2,\bf{c}_3,\bf{c}_4,\bf{u},D)$, do:
\begin{enumerate}
	\item Compute $\bf{d}= H'(\bf{c}_1\|\bf{c}_2\|\bf{c}_3\|\bf{c}_4)\in\{0,1\}^l$ and check whether $\|\bf{d}\|_1\leq w$. If not then return $\perp$, otherwise continue to Step 2.
	\item Check whether $A\bf{u}=D\bf{d}\mod q$. If not then returns $\perp$;  Otherwise, continue to Step 3.
	\item Compute $\bf{b}=H'(\bf{c}_1\|\bf{c}_2\|D)\in\{-1,1\}^l$ and sample $\bf{e}\in\ZZ^{2m\times t}$ from $$\bf{e}\gets\SampleLeft(A, B+\sum_{i=1}^lb_iA_i,T_A,U,\sigma).$$
	Note that $F_1\cdot\bf{e}=U$ in $\ZZ^{n\times t}_q$.
	\item Compute $\bf{w}\gets\bf{c}_1-\bf{e}^T\bf{c}_3\in\ZZ_q^t$.
	\item For each $i=1,\cdots, t$, compare $w_i$ and $\lfloor\frac{q}{2}\rfloor$. If they are close, output $m_i=1$ and otherwise output $m_i=0$. We then obtain the message $\bf{m}$.
	\item Sample $\bf{e}'\in\ZZ^{2m\times t}$ from
	$$\bf{e}'\gets\SampleLeft(A', B+\sum_{i=1}^lb_iA_i,T_{A'},U,\sigma).$$
	\item Compute $\bf{w}'\gets\bf{c}_2-(\bf{e}')^T\bf{c}_4\in\ZZ_q^t$.
	\item For each $i=1,\cdots,t$, compare $w'_i$ and $\lfloor\frac{q}{2}\rfloor$. If they are close, output $h_i=1$ and otherwise output $h_i=0$. We then obtain the vector $\bf{h}$.
	\item If $\bf{h}=H(\bf{m})$ then output $\bf{m}$, otherwise output $\perp$.
\end{enumerate}
\medskip
\noindent Trapdoor($\SK_i$): On input a user $U_i$'s secret key $\SK_i=(K_{i,1}, K_{i,2})$, it outputs a trapdoor $\td_i=K_{i,2}$.

\medskip
\noindent Test($\td_i,\td_j,\CT_i,\CT_j$): On input trapdoors $\td_i, \td_j$ and ciphertexts $\CT_i,\CT_j$ for users $U_i, U_j$ respectively, computes 
\begin{enumerate}
	\item For each $i$ (resp. $j$), do the following:
	\begin{itemize}
		\item Compute $$\bf{b}_i=H'(\bf{c}_{i1}\|\bf{c}_{i2}\|D_i) =(b_{i1},\cdots,b_{il})$$ and sample $\bf{e}_i\in\ZZ^{2m\times t}$ from
		$$\bf{e_i}\gets\SampleLeft(A'_i, B_i+\sum_{k=1}^lb_{ik}A_{ik},T_{A'_i},U_i,\sigma).$$
		Note that $F_{i2}\cdot\bf{e}_i=U_i$ in $\ZZ^{n\times t}_q$.
		\item Compute $\bf{w}_i\gets \bf{c_{i2}}-\bf{e}_i^T\bf{c}_{i4}\in\ZZ_q^t$. For each $k=1,\cdots, t$, compare each coordinate $w_{ik}$ with $\lfloor\frac{q}{w}\rfloor$ and output $\bf{h}_{ik}=1$ if they are close, and $0$ otherwise. At the end, we obtain the vector $\bf{h}_i$ (resp. $\bf{h}_j$). 
	\end{itemize}
	\item Output $1$ if  $\bf{h}_i=\bf{h}_j$ and $0$ otherwise.
\end{enumerate}

\begin{thm}
	Our PKEET construction above is correct if $H$  is a collision-resistant hash function.
\end{thm}
\begin{proof}
	It is easy to see that if $\CT$ is a valid ciphertext of $\bf{m}$ then the decryption will always output $\bf{m}$. Moreover, if $\CT_i$ and $\CT_j$ are valid ciphertext of $\bf{m}$ and $\bf{m}'$ of user $U_i$ and $U_j$ respectively. Then the Test process checks whether $H(\bf{m})=H(\bf{m}')$. If so then it outputs $1$, meaning that $\bf{m}=\bf{m}'$, which is always correct with overwhelming probability since $H$ is collision resistant. Hence our PKEET described above is correct.
\end{proof}

\subsection{Parameters}\label{sec:params}
We follow~{\cite[Section 7.3]{ABB10-EuroCrypt}} and~{\cite[Section 4.1]{LM18}} for choosing parameters for our scheme. Now for the system to work correctly we need to ensure
\begin{itemize}
	\item the error term in decryption is less than ${q}/{5}$ with high probability, i.e., $q=\Omega(\sigma m^{3/2})$ and $\alpha<[\sigma lm\omega(\sqrt{\log m})]^{-1}$,
	\item that the $\TrapGen$ can operate, i.e., $m>6n\log q$,
	\item that $\sigma$ is large enough for $\SampleLeft$ and $\SampleRight$, i.e.,
	$\sigma>lm\omega(\sqrt{\log m})$,
	\item that Regev's reduction applies, i.e., $q>2\sqrt{n}/\alpha$,
	\item that our security reduction applies (i.e., $q>2Q$ where $Q$ is the number of identity queries from the adversary).
\end{itemize}
Hence the following choice of parameters $(q,m,\sigma,\alpha)$ from \cite{ABB10-EuroCrypt} satisfies all of the above conditions, taking $n$ to be the security parameter:
\begin{equation}\label{eq:params}
	\begin{aligned}
		& m=6n^{1+\delta}\quad,\quad q=\max(2Q,m^{2.5}\omega(\sqrt{\log n})) \\
		& \sigma = ml\omega(\sqrt{\log n})\quad,\quad\alpha=[l^2m^2\omega(\sqrt{\log n})]\\
		&b=1, \quad,\quad w =O(n/\log(n)) 
	\end{aligned}
\end{equation}
and round up $m$ to the nearest larger integer and $q$ to the nearest larger prime. Here we assume that $\delta$ is such that $n^\delta>\lceil\log q\rceil=O(\log n)$. This choice of parameters~\eqref{eq:params} will also ensure the security of the one-time signature scheme in Section~\ref{sec:one-time signature}; cf.~Theorem~\ref{sig strong SIS}.
\subsection{Security analysis}
The proposed scheme is OW-CCA2 secure against Type-I adversaries (cf.~Theorem~\ref{thm:OW}) and IND-CCA2 secure against Type-II adversaries (cf. Theorem~\ref{thm:IND}). The proofs are similar to those in~\cite{Duong19} and those in Section~\ref{sec:construction-ideal} and hence we omit the full proofs. Note that in order to answer decryption queries and achieve CCA2-security, we simply follow the standard procedure as in~\cite{CHK04}.

\begin{thm}\label{thm:OW}
	The PKEET with parameters $(q,n,m,\sigma,\alpha)$ as in~\eqref{eq:params} is $\OWCCA$ secure provided that $H$ is a one-way hash function, $H'$ is a collision-resistant hash function.
	
\end{thm}

\begin{thm}\label{thm:IND}
	The PKEET with parameters $(q,n,m,\sigma,\alpha)$ as in~\eqref{eq:params}  is $\INDCCA$ secure provided that $H'$ is a collision-resistant hash function.
\end{thm}

\section{Conclusion}
In this paper, we propose a direct construction of an efficient PKEET scheme over ideal lattices. We also propose a modification of the PKEET construction over integer lattices from~\cite{Duong19}. Our two schemes are proven to be CCA2-secure in the standard model. We also provide an instantiation from the generic construction by Lee et al.~\cite{Lee2016} from which we conclude that our schemes are much more efficient than that generic construction. It is an interesting question of whether one can further improve the efficiency of PKEET constructions while still mantaining the CCA2 security.



\section*{Appendix A: An instantiation of Lee et al.'s construction}
In this section, we will present a lattice-based PKEET which is an instantiation of the Lee et al.'s construction~\cite{Lee2016}. In their generic construction, they need (i) a multi-bit HIBE scheme and (ii) an one-time signature scheme. To instantiate their construction, we modify the lattice based single-bit HIBE of \cite{ABB10-EuroCrypt} to multi-bit one and use it, along with the signature scheme, to have following construction of lattice based PKEET.  Even though one needs only a one-time signature scheme, we choose the full secure signature scheme from~\cite{ABB10-EuroCrypt} to unify the system, since in such case, both signature and HIBE schemes use the same public key.
It is required to use multi-bit HIBE and signature scheme to have PKEET from Lee et al.'s \cite{Lee2016}.

In what follows, we will denote by $[id_1.id_2]$ the identity of a $2$-level HIBE scheme where $id_1$ is the first level identity and $id_2$ is the second level identity. Below, we follow~\cite{Lee2016} to denote by $[\ID.0]$ (resp.  $[\ID.1]$) an identity in the second level in which we indicate that $\ID$ is the identity of the first level.
\subsection{Construction}\label{sec:construct Lee}
\begin{description}
	\item[$\Setup$($\lambda$)]$\\$ On input security parameter $\lambda$, and a maximum hierarchy depth $2$, set the parameters $q, n, m, {\bar \sigma}, {\bar \alpha}$. The vector ${\bar \sigma} ~\&~ {\bar \alpha} \in \mathbb{R}^2$ and we use $\sigma_l$ and $\alpha_l$ to refer to their $l$- th coordinate. 
	\begin{enumerate}
		\item Use algorithm $\TrapGen(q, n)$ to select a uniformly random $n \times m$- matrix $A \in \mathbb{Z}_q^{n \times m}$ with a basis $T_{A}$ for $\L^{\perp}_q (A)$ and $\L^{\perp}_q (A')$, respectively. Repeat this Step until $A$ has rank $n$.
		\item Select uniformly random $m \times m$ matrices $A_1,A_2,B \in \mathbb{Z}_q^{n \times m}$.
		\item Select a uniformly random matrix $U\in\mathbb{Z}_q^{n\times t}$.
		\item We need some hash functions $H: \{0, 1\}^* \rightarrow\{0,1\}^t$,  $H_1: \{0, 1\}^* \rightarrow \{-1, 1\}^t$, $H_2 :\{0,1\}^*\to\ZZ^n_q$ and a full domain difference map $H' :\ZZ_q^n\to\ZZ_q^{n\times n}$ as in~{\cite[Section 5]{ABB10-EuroCrypt}}.
		\item Output the public key and the secret key
		$$\PK=(A,A_1,A_2 ,B,U)\quad,\quad\SK=T_A .$$
	\end{enumerate}

	\item[$\Enc(\PP, \ID, \bf{m})$] $\\$
	On input the public key $\PK$ and a message $\bf{m}\in\{0,1\}^t$ do
	\begin{enumerate}
		\item Choose uniformly random $\bf{s}_1,\bf{s}_2\in\ZZ_q^n$.
		\item Choose $\bf{x}_1,\bf{x}_2\in\overline{\Psi}_\alpha^t$ and compute
		\begin{align*}
		\bf{c}_1 &= U^T\bf{s}_1 +\bf{x}_1 +\bf{m}\big\lfloor\frac{q}{2}\big\rfloor\in\ZZ_q^t,\\
		\bf{c}_2 &= U^T\bf{s}_2 +\bf{x}_2 +H(\bf{m})\big\lfloor\frac{q}{2}\big\rfloor \in\ZZ_q^t.
		\end{align*}
	\item Choose $K\in\ZZ_q^{m\times k}$ uniformly at random and compute $D=AK\in\ZZ_q^{n\times k}$.
		\item Set 	$\ID := H_2(D)\in\ZZ_q^n$.
		\item Build the following matrices in $\ZZ_q^{n\times 3m}$:
		\begin{align*}
		F_{\ID.0} &= (A | A_1 + H'(0)\cdot B | A_2 + H'(\ID)\cdot B),\\
		F_{\ID.1} &= (A | A_1 + H'(1)\cdot B | A_2 + H'(\ID)\cdot B).
		\end{align*}
		
		\item Choose a uniformly random $n\times 2m$ matrix $R$ in $\{-1,1\}^{n\times 2m}$.
		\item Choose $\bf{y}_1, \bf{y}_2\in\overline{\Psi}_\alpha^m$ and set $\bf{z}_1=R^T\bf{y}_1, \bf{z}_2=R^T\bf{y}_2\in\ZZ_q^{2m}$.
		\item Compute
		\begin{align*}
		\bf{c}_3&=F_{\ID.0}^T\bf{s}_1+[\bf{y}_1^T|\bf{z}_1^T]^T\in\ZZ_q^{3m},\\ 
		\bf{c}_4&=F_{\ID.1}^T\bf{s}_2+[\bf{y}_2^T|\bf{z}_2^T]^T\in\ZZ_q^{3m}.
		\end{align*}
		
		\item Let $\bf{b}:= H_1(\bf{c}_1\|\bf{c}_2\|\bf{c}_3\|\bf{c}_4)\in\{0,1\}^k$ such that $\|\bf{b}\|_1\leq w$.
		\item Compute $\bf{u}=K\bf{b}\in\ZZ_q^m$.
		\item Output the ciphertext 
		$$\CT=(\bf{c}_1,\bf{c}_2,\bf{c}_3,\bf{c}_4,\bf{u},D)\in\ZZ_q^{2t+7m}\times\ZZ_q^{n\times k} .$$
	\end{enumerate}

	\item[$\Dec (\PP, \SK,\CT)$]$\\$ On input a secret key $\SK_{\ID}$ and a ciphertext $\CT$, do 
	\begin{enumerate}
		\item Parse the ciphertext $\CT$ into $(\bf{c}_1,\bf{c}_2,\bf{c}_3,\bf{c}_4,\bf{u},D)$ and compute $\bf{d}= H'(\bf{c}_1\|\bf{c}_2\|\bf{c}_3\|\bf{c}_4)\in\{0,1\}^l$ and check whether $\|\bf{d}\|_1\leq w$. If not then return $\perp$, otherwise continue to Step 2.
		\item Check whether $A\bf{u}=D\bf{d}\mod q$. If not then returns $\perp$;  Otherwise, continue to Step 3.
		\item Let $\bf{b}:= H_1(\bf{c}_1\|\bf{c}_2\|\bf{c}_3\|\bf{c}_4)\in\{-1,1\}^l$ and define a matrix
		$$F=(D| B+\sum_{i=1}^lb_iA_i)\in\ZZ_q^{n\times 2m}.$$
		\item If $F\cdot\bf{u}=0$ in $\ZZ_q$ and $\|\bf{e}\|\leq\sigma\sqrt{2m}$ then continue to Step 5; otherwise output $\perp$.
		\item Set $\ID:= H_2(D)\in\ZZ_q^n$ and build the following matrices:
		\begin{align*}
		F_{\ID.0} &= (A | A_1 + H'(0)\cdot B| A_2 + H'(\ID)\cdot B) \in \ZZ_q^{n\times 3m},\\
		F_{\ID.1} &= (A | A_1+ H'(1)\cdot B| A_2 + H'(\ID)\cdot B) \in \ZZ_q^{n\times 3m}.
		\end{align*}
		\item Generate
		\begin{align*}
		E_{0} \gets \SampleBasisLeft(A, A_1 + H'(0)\cdot B, T_A, \sigma)\\
		E_{1} \gets \SampleBasisLeft(A, A_1 + H'(0)\cdot B, T_A, \sigma)\\ 
		E_{\ID.0} \gets \SampleLeft(A|A_1 + H'(0)\cdot B, A_2 + H'(\ID)\cdot B, E_{0}, U, \sigma)\\ ~s.t.~ F_{\ID.0} \cdot E_{\ID.0} = U\mod q\\
		E_{\ID.1} \gets \SampleLeft(A|A_1 + H'(1)\cdot B, A_2 + H'(\ID)\cdot B, E_{1}, U, \sigma)\\ ~s.t.~ F_{\ID.1} \cdot E_{\ID.1} = U\mod q.
		\end{align*}
		\item Compute $\bf{w}\gets\bf{c}_1-E_{\ID.0}^T\bf{c}_3\in\ZZ_q^t$.
		\item For each $i=1,\cdots, t$, compare $w_i$ and $\lfloor\frac{q}{2}\rfloor$. If they are close, output $m_i=1$ and otherwise output $m_i=0$. We then obtain the message $\bf{m}$.
		\item Compute $\bf{w}'\gets\bf{c}_2-E_{\ID.1}^T\bf{c}_4\in\ZZ_q^t$.
		\item For each $i=1,\cdots,t$, compare $w'_i$ and $\lfloor\frac{q}{2}\rfloor$. If they are close, output $h_i=1$ and otherwise output $h_i=0$. We then obtain the vector $\bf{h}$.
		\item If $\bf{h}=H(\bf{m})$  then output $\bf{m}$, otherwise output $\perp$.
	\end{enumerate}
	
	\item[$\Td(\SK_i)$]$\\$
	On input the secret key $\SK_i (= T_{A_i})$ of a user $U_i$, then generates $\td_i:=E_{i,1}$ as in Step 6 in Decryption process, i.e.,
	$$\td_i\gets \SampleBasisLeft(A_i, A_{i,1} + H'(0)\cdot B_i, T_{A_i}, \sigma)$$

	\item[$\Test(\td_i,\td_j,\CT_i,\CT_j)$]$\\$ On input trapdoors $\td_i,\td_j$ and ciphertexts $\CT_i,\CT_j$ of users $U_i$ and $U_j$ respectively, for $k=i,j$, do the following
	\begin{enumerate}
		\item Parse $\CT_k$ into 
		$$(\bf{c}_{k,1},\bf{c}_{k,2},\bf{c}_{k,3},\bf{c}_{k,4},vk_k,\bf{u}_k,D_k).$$
		\item 
		Sample $E_{\ID_k.1}\in\ZZ_q^{3m\times t}$ from
		$$ E_{\ID_k.1} \gets \SampleLeft(A_k|A_{k,1} + H'(1)\cdot B_k, A_{k,2} + H'(\ID)\cdot B_k, \td_k, U, \sigma)$$
		\item Use $E_{\ID_k.1}$ to decrypt $\bf{c}_{k,2}$, $\bf{c}_{k,4}$ as in Steps 9-10 of $\Dec(\SK,\CT)$ above to obtain the hash value $\bf{h}_k$.
		\item If $\bf{h}_i=\bf{h}_j$ then ouput $1$; otherwise output $0$.
	\end{enumerate}
\end{description}

\begin{thm}[Correctness]
	The above PKEET is correct if the hash function $H$ is collision resistant.
\end{thm}
\begin{proof}
	Since we employ the multi-bit HIBE and signature scheme from~\cite{ABB10-EuroCrypt}, their correctness follow from~\cite{ABB10-EuroCrypt}. The Theorem follows from~{\cite[Theorem 1]{Lee2016}}.
\end{proof}

\subsection{Parameters}
We follow~{\cite[Section 8.3]{ABB10-EuroCrypt}} for choosing parameters for our scheme. Now for the system to work correctly we need to ensure
\begin{itemize}
	\item the error term in decryption is less than ${q}/{5}$ with high probability, i.e., $q=\Omega(\sigma m^{3/2})$ and $\alpha<[\sigma lm\omega(\sqrt{\log m})]^{-1}$,
	\item that the $\TrapGen$ can operate, i.e., $m>6n\log q$,
	\item that $\sigma$ is large enough for $\SampleLeft$ and $\SampleRight$, i.e.,
	$\sigma>lm\omega(\sqrt{\log m})$,
	\item that Regev's reduction applies, i.e., $q>2\sqrt{n}/\alpha$,
\end{itemize}
Hence the following choice of parameters $(q,m,\sigma,\alpha)$ from \cite{ABB10-EuroCrypt} satisfies all of the above conditions, taking $n$ to be the security parameter:
\begin{equation}\label{eq:params-HIBE}
\begin{aligned}
& m=6n^{1+\delta}\quad,\quad q=\max(2Q,m^{2.5}\omega(\sqrt{\log n})) \\
& \sigma = ml\omega(\sqrt{\log n})\quad,\quad\alpha=[l^2m^2\omega(\sqrt{\log n})]
\end{aligned}
\end{equation}
and round up $m$ to the nearest larger integer and $q$ to the nearest larger prime. Here we assume that $\delta$ is such that $n^\delta>\lceil\log q\rceil=O(\log n)$.

\begin{thm}
	The PKEET constructed in Section~\ref{sec:construct Lee} with paramaters as in~\eqref{eq:params-HIBE} is $\textsf{IND-ID-CCA2}$ secure provided that $H_1$ is collision resistant. 
\end{thm}

\begin{proof}
	The HIBE is $\textsf{IND-sID-CPA}$ secure by~{\cite[Theorem 33]{ABB10-EuroCrypt}} and the signature is strongly unforgeable by Theorem~\ref{sig strong SIS}. The result follows from {\cite[Theorem 5]{Lee2016}}.
\end{proof}

\begin{thm}[{\cite[Theorem 3]{Lee2016}}]
	The PKEET with parameters $(q,n,m,\sigma,\alpha)$ as in~\eqref{eq:params-HIBE} is $\textsf{OW-ID-CCA2}$ provided that $H$ is one-way and  $H_1$ is collision resistant.
\end{thm}

\begin{proof}
	The HIBE is $\textsf{IND-sID-CPA}$ secure by~{\cite[Theorem 33]{ABB10-EuroCrypt}} and the signature is strongly unforgeable by~Theorem~\ref{sig strong SIS}. The result follows from {\cite[Theorem 6]{Lee2016}}.
\end{proof}

\end{document}